\documentclass[12pt]{article}
\pdfoutput=1
\usepackage{graphicx}
\usepackage{mathptmx}
\usepackage{mathrsfs}
\usepackage{amssymb}
\usepackage{amsmath}
\usepackage{graphicx}
\usepackage{amsmath,amsfonts}
\usepackage{listings}
\usepackage[square,numbers]{natbib}
\newtheorem{lemma}{Lemma}[section]
\newtheorem{theorem}{Theorem}[section]
\newtheorem{definition}{Definition}[section]

\newtheorem{remark}{Remark}[section]
\newenvironment{proof}{\noindent\\ \noindent\relax{\sc
     Proof}}{{\samepage\par\nopagebreak\hbox
     to\hsize{\hfill$\Box$}}}
\newcommand{\be}{\begin{equation}} \newcommand{\ee}{\end{equation}}
\newcommand{\bd}{\begin{displaymath}} \newcommand{\ed}{\end{displaymath}}
\newcommand{\ba}{\begin{align}} \newcommand{\ea}{\end{align}}
\newcommand{\baa}{\begin{align*}} \newcommand{\eaa}{\end{align*}}
\newcommand{\ben}{\begin{enumerate}} \newcommand{\een}{\end{enumerate}}
\newcommand{\bi}{\begin{itemize}} \newcommand{\ei}{\end{itemize}}

\newcommand{\ud}{\mathrm{d}}
\newcommand{\Expectation}[1]{\operatorname{E}\left[ #1 \right]}
\newcommand{\Var}[1]{\operatorname{Var}\left[ #1 \right]}

\newcommand{\covariance}[2]{\operatorname{Cov}\left[ #1,#2 \right]}

\newcommand{\Et}[1]{\operatorname{E}\left[ #1 \vert \mathcal{Y}_{n}\right]}
\newcommand{\Vart}[1]{\operatorname{Var}\left[ #1 \vert  \mathcal{Y}_{n} \right]}
\newcommand{\covt}[2]{\operatorname{Cov}\left[ #1,#2 \vert  \mathcal{Y}_{n}\right]}

\begin{document}

\title{Phylogenetic confidence intervals for the optimal trait value}
\author{{\sc Krzysztof Bartoszek} and {\sc Serik Sagitov}} 

\maketitle

\begin{abstract}
We consider a stochastic evolutionary  model for a phenotype developing amongst $n$
related species with unknown phylogeny.
The unknown tree is modelled by a Yule process conditioned on $n$
contemporary nodes.
The trait value is assumed to evolve along lineages as an Ornstein--Uhlenbeck process. 
As a result, the trait values of  the $n$ species form a sample with dependent observations.
We establish three limit theorems for the sample mean corresponding to three domains for  the
adaptation rate.
In the case of fast adaptation, we show that for large $n$ the normalized sample mean is
approximately normally distributed.
Using these limit theorems, we develop novel confidence interval formulae for the optimal 
trait value.
\end{abstract}

Keywords : 
Central limit theorem, 
Conditioned Yule process,
macroevolution, martingales, 
Ornstein--Uhlenbeck process, phylogenetics

\section{Introduction}\label{intro}

Phylogenetic comparative methods deal with  multi-species trait value data.
This is an established and rapidly expanding area of research concerning evolution of phenotypes in groups of
related species living under various environmental conditions. An important feature of such data is the 
branching structure of evolution causing dependence among the observed trait values. 
For this reason the usual starting point for phylogenetic comparative studies is an inferred phylogeny 
describing the evolutionary relationships. The likelihood can be computed by assuming a model for trait 
evolution along the branches of this fixed tree, such as the Ornstein-Uhlenbeck process.

The one-dimensional Ornstein-Uhlenbeck model is characterized by four parameters: the optimal value $\theta$, 
the adaptation rate $\alpha>0$,  the ancestral value $X_0$, and the noise size $\sigma$. 
The classical Brownian motion model \cite{JFel85} can be viewed as a special case with $\alpha=0$ and $\theta$ 
being irrelevant. As with any statistical procedure, it is important to be able to compute confidence  
intervals for these parameters. However, confidence intervals are often not mentioned in 
phylogenetic comparative studies \cite{CBoeGCooPRal2012}. 

There are a number of possible numerical ways of calculating such confidence intervals when the underlying 
phylogenetic tree is known. 
Using a regression framework one can apply standard regression theory methods to compute confidence 
intervals for $(\theta,X_0)$ conditionally on $(\alpha,\sigma^2)$ \cite{TGarAIve2000,THan1997,EMarTHan1997,FRoh2001}.
Notably in \cite{TGarPMidAIve1999} the authors derive analytical formulae for confidence intervals for $X_0$ under the 
Brownian motion model.
In more complicated situations a parametric bootstrap is a (computationally very demanding) way out   
\cite{CBoeGCooPRal2012,MButAKinOUCH, AIvePMidTGar2007}. Another approach
is to report a support surface \cite{THan1997,THanJPieSOrzSLOUCH}, or consider the curvature of the 
likelihood surface  \cite{KBaretal}.

All of the above methods have in common that they assume that the phylogeny describing the evolutionary 
relationships is fully resolved.
Possible errors in the topology can cause problems -- the closer to the tips they occur, the more problematic they can be
\cite{MSym2002}. 
On the other hand, the regression estimators will remain unbiased even with a misspecified
tree \cite{FRoh2006} and also seem to be robust with respect to errors in the phylogeny
at least for the Brownian motion model \cite{ESto2011}. There are only few papers addressing the issue of 
phylogenetic uncertainty.
An MCMC procedure to jointly estimate the phylogeny and parameters of the 
Brownian model of trait evolution was suggested in \cite{JHueBRanJMas2000,JHueBRan}.
Recently, \cite{GSlaetal2012} develops an Approximate Bayesian Computation framework to estimate
Brownian motion parameters in the case of an incomplete tree.

Our paper studies a situation when nothing is known about the phylogeny. The simplest stochastic model addressing 
this case is a combination of a Yule tree and the Brownian motion on top of it: 
already in the 1970s, a 
joint maximum likelihood estimation procedure of a Yule tree and Brownian motion on top of 
it was proposed 
in \cite{AEdw1970}. This basic evolutionary model allows for far reaching analytical analysis  
\cite{KBar2014,FCraMSuc2013,SSagKBar2012}.
A more realistic stochastic model of this kind combines the Brownian motion with a birth-death tree allowing for 
extinction of species  \cite{FBok2010}. For the latter model \cite{SSagKBar2012} explicitly compute the 
so-called 
interspecies correlation coefficient.
Such ``tree-free" models are appropriate for working with fossil data when there may be available rich fossilized 
phenotypic information but the 
molecular material might have degraded so much that it is impossible to infer evolutionary relationships. 
In \cite{FCraMSuc2013}  the usefulness of the tree-free approach for contemporary species 
is demonstrated in an Carnivora order case study and in \cite{WMulFCra2015} the distribution
over the space of Yule trees of the interspecies correlation coefficient is calculated.

Conditioned birth--death processes as stochastic models for species trees,
have received significant attention in the last decade 
\cite{DAldLPop2005,TGer2008a,AMooetal,TreeSim1,TreeSim2,TStaMSte2012}.
In this work the unknown tree is modeled by the Yule process conditioned on $n$ extant species while the 
evolution of a trait along a lineage is viewed as the Ornstein-Uhlenbeck process, see Fig. \ref{tr}. 
We study the properties of the sample mean and sample variance computed from the vector of $n$ trait values. 
Our main results are three asymptotic confidence interval formulae for the optimal trait value $\theta$. 
These three formulae represent three asymptotic regimes for different values of the adaptation rate $\alpha$.

In the discussion in \cite{FCraMSuc2013} it is pointed out that ``as evolutionary biologists further refine our 
knowledge of
the tree of life, the number of clades whose phylogeny is truly unknown may diminish, along with
interest in tree-free estimation methods." In our opinion the main contribution
of such methods is that they indicate statistical and asymptotic properties of phylogenetic samples
under given evolutionary models. These properties can then be verified for 
other models of tree growth or real phylogenies
\cite{CAne2008,CAneLHoSRoc2014,OGasMSte2014,LHoCAne2013,LHoCAne2014,LHoCAne2015,EMosMSte2014}.
We believe furthermore that  
the easy-to-compute tree-free predictions 
will always play an important role of a sanity check to see
whether the conclusions based on the inferred phylogeny deviate much from those
from a ``typical'' phylogeny. 
Moreover, results like those presented here can also be used as a method of testing 
software for phylogenetic comparative models.

A detailed description of the evolutionary model along with our 
main results are presented in Section 
\ref{Smain}. Section \ref{Syou} contains new formulae for the Laplace transforms of important 
characteristics of the conditioned Yule species tree: the time to origin $U_{n}$ and the time $\tau^{(n)}$ to the most 
recent common ancestor for a pair of two species chosen at random out of $n$ extant species.
In Section \ref{Sout} we calculate the interspecies correlation coefficient for the Yule--Ornstein--Uhlenbeck model
and Section \ref{SecMart} contains the proof of our limit theorems. 
In Section \ref{Scons} we establish the consistency of the stationary variance estimator, 
which is needed for our confidence interval formulae, cf \cite{THan1997} where 
the residual sum of squares was suggested to estimate the stationary variance.
In Appendix \ref{appSaM} we calculate all the joint moments of $U_{n}$ and $\tau^{(n)}$.

Our main result, Theorem \ref{tMart}, should be compared with the limit theorems obtained in
\cite{RAdaPMil20111,RAdaPMil20112}. They also revealed  three asymptotic regimes in a related, 
though different setting, 
dealing with a branching Ornstein--Uhlenbeck process. In their case the time of observation is 
deterministic and the number of the tree tips is random, while in our case the observation time is 
random and the number of the tips is deterministic. Although it is possible (with some effort) 
to deduce our results from  \cite{RAdaPMil20111,RAdaPMil20112}, our proof provides a much more 
elementary derivation. We believe that our approach will be useful in addressing other biologically 
relevant issues like the formulae for the higher moments given in Appendix \ref{appSaM}. 
Another similar limit theorem, but one conditional on the sequence of species trees
generated by different mechanisms,
is derived in \cite{CAneLHoSRoc2014}.

\begin{figure}
\centering
\includegraphics[width=10cm]{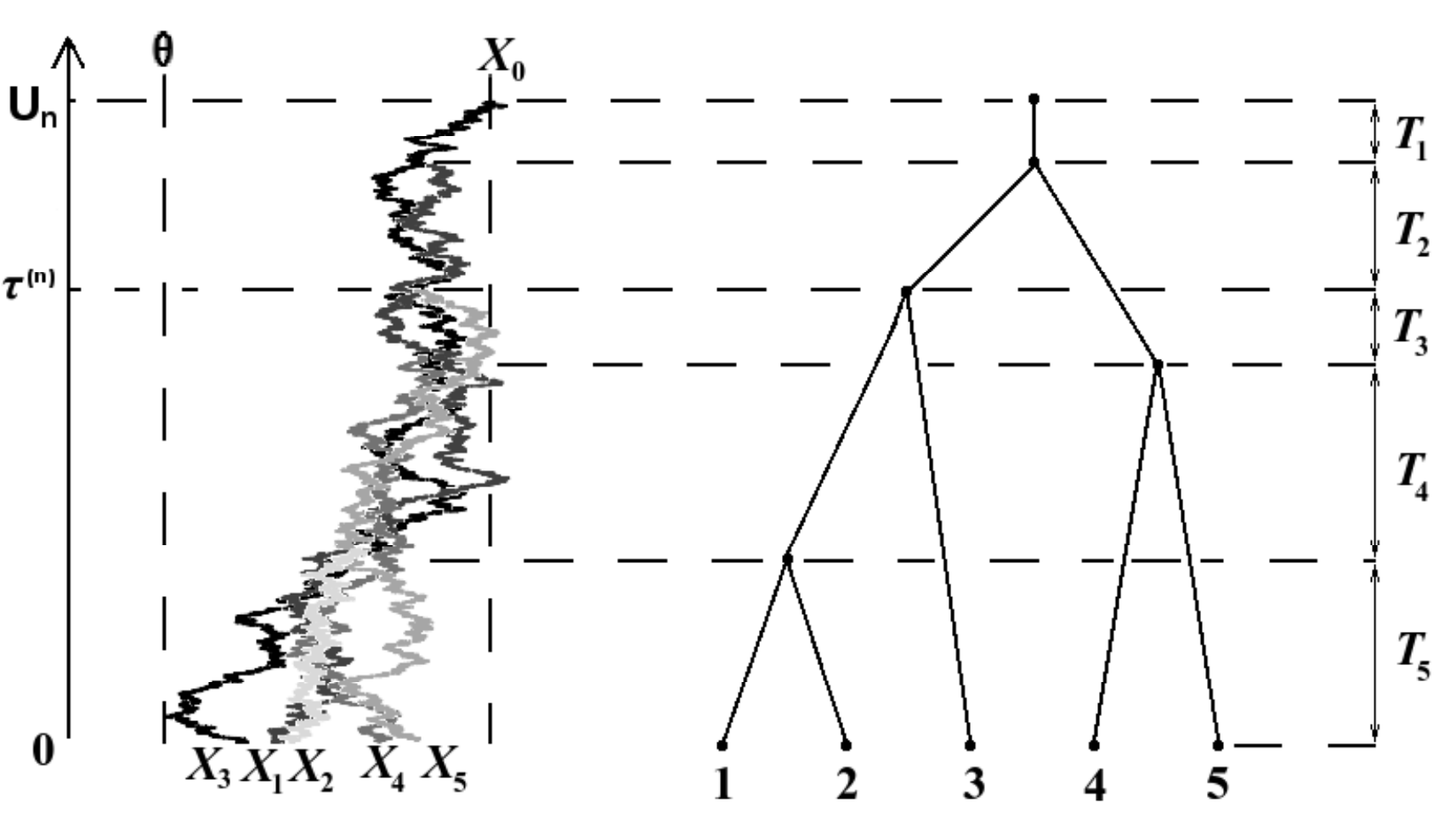}
\caption{
On the left: a branching Ornstein--Uhlenbeck process simulated on a realization of the  Yule $n$-tree with $n=5$ tips using the 
TreeSim \cite{TreeSim1,TreeSim2} and mvSLOUCH \cite{KBaretal} R \cite{R} packages. 
Parameters used are $\alpha=1$, $\sigma=1$,  $X_0-\theta=2$,  after the tree height $U_{n}$ was scaled to 1.
On the right: the species tree disregarding the trait values supplied with the notation for the inter--speciation times. 
For the pair of tips (2,3) the time $\tau^{(n)}$ to their most recent common ancestor is marked on the time axis 
(starting at present and going back to the time of origin).}
\label{tr}
\end{figure}

\section{The model and  main results}\label{Smain}

This work deals with what we call the Yule--Ornstein--Uhlenbeck model which is characterized by four parameters $(X_0,\alpha,\sigma,\theta)$ and consists of two ingredients
\begin{enumerate}
\item the species tree connecting $n$ extant species is modeled by the pure birth Yule process \cite{GYul1924} with 
a unit speciation rate $\lambda=1$ and conditioned on having $n$ tips  \cite{TGer2008a}, 
\item the observed trait values $(X_1^{(n)},\ldots,X_n^{(n)})$ on the tips of the tree evolved from the ancestral state $X_0$ following the 
Ornstein--Uhlenbeck process with parameters $(\alpha,\sigma,\theta)$. 
\end{enumerate}
\begin{definition} \label{def}
Let $(T_{1},\ldots,T_{n})$ be independent exponential random variables  with parameters $(1,\ldots,n)$.  We define the Yule $n$-tree as  a random tree with $n$  tips which is
 constructed using a bottom-up algorithm based on the following two simple rules.
 
(1) During the time period $T_k$ the tree has $k$ branches. 

(2) For $k\in[2,n]$ the reduction from $k$ to $k-1$ branches occurs as two randomly chosen branches 
coalesce into one branch. 

\noindent The height the Yule $n$-tree is now $U_n=T_1+\ldots+T_n$.
\end{definition}
As shown in \cite{TGer2008a}, this definition corresponds to the standard Yule tree conditioned on having  $n$  tips at the moment of observation, assuming that the time to the origin has the improper uniform prior.

Following \cite{MButAKinOUCH,THan1997}, we model  trait evolution along a lineage using the Ornstein--Uhlenbeck process  $X(t)$ given by the stochastic differential equation

\be\label{eqOU}
\ud X(t) = -\alpha(X(t)-\theta)\ud t + \sigma \ud B(t),\quad X(0)=X_{0}.
\ee  
Here  $\alpha > 0$ is the adaptation rate,  $\theta$ is the optimal trait value, $\sigma^2$ is the noise variance, and $B(t)$ is the standard Wiener process. 
The distribution of $X(t)$ is normal with,

\be\label{solOU}
\Expectation{X(t)}=\theta + e^{-\alpha t}(X_{0}-\theta), \quad \Var{X(t)}=\frac{\sigma^{2}}{2\alpha}(1-e^{-2 \alpha t}),
\ee
implying that  $X(t)$ looses the effect of the ancestral state $X_{0}$ at an exponential rate. 
In the long run the Ornstein--Uhlenbeck process acquires a stationary normal distribution with mean $\theta$ and variance $\sigma^{2}/2\alpha$.

We propose  asymptotic confidence interval formulae for the optimal value $\theta$  which take into account phylogenetic uncertainty. To this end we study properties of the sample mean and sample variance

\begin{displaymath}
\overline{X}_{n}={X_1^{(n)}+\ldots+X_n^{(n)}\over n},\quad S^{2}_{n}=\frac{1}{n-1}\sum_{i=1}^{n}(X_{i}^{(n)}-\overline{X}_{n})^{2}.
\end{displaymath}
Using the properties of the Yule--Ornstein--Uhlenbeck model we find explicit expressions for $\Expectation{\overline X_n}$, $\Var{\overline X_n}$,  $\Expectation{S_n^2}$,
study the asymptotics of  $\Var{S_n^2}$, and prove the following limit theorem revealing three different asymptotic regimes.

\begin{theorem}\label{tMart}
Let $\delta={X_0-\theta\over\sqrt{\sigma^{2}/2\alpha}}$ be a normalized difference between the ancestral and optimal values.
Consider the normalized sample mean 
$\overline{Y}_{n}=\frac{\overline{X}_{n}-\theta}{\sqrt{\sigma^{2}/2\alpha}}$ of the Yule-Ornstein-Uhlenbeck process  with $\overline{Y}_0=\delta$. As $n\to\infty$ the process $\overline{Y}_{n}$ has the following limit behavior.

(i) If $\alpha>0.5$, then $\sqrt{n}\cdot \overline{Y}_{n}$ is asymptotically normally distributed with zero mean  and variance 
$\frac{2\alpha+1}{ 2\alpha-1}$.

(ii)  If $\alpha=0.5$, then  $\sqrt{n/\ln n}\cdot \overline{Y}_{n}$ is asymptotically normally distributed with zero mean  and variance $2$.

(iii) If $\alpha<0.5$, then
$n^{\alpha}\cdot \overline{Y}_{n}$ converges a.s. and in $L^2$ to a random variable $Y_{\alpha,\delta}$ with  $\Expectation{Y_{\alpha,\delta}}=\delta\Gamma(1+\alpha)$ and $\Expectation{Y_{\alpha,\delta}^2}=\left(\delta^{2}+\frac{4\alpha}{1-2\alpha} \right)\Gamma(1+2\alpha)$.
\end{theorem}
Let $z_{x}$ be the $x$-quantile of the standard normal distribution, and  $q_x$ be  the $x$-quantile of the limit $Y_{\alpha,\delta}$. Denote by $S_n$ the sample standard deviation defined as the square root of $S_n^2$. As it will be shown in Section \ref{Scons}, the sample variance $S_n^2$ is a consistent estimator of  ${\sigma^2\over2\alpha}$. This fact together with Theorem \ref{tMart} allows us to state the following three approximate 
$(1-x)$-level confidence intervals for $\theta$ assuming that we know the value of $\alpha$:

\begin{align*}
\mbox{for } \alpha>0.5\qquad   & \overline{X}_{n} \pm z_{1-x/2}\cdot S_{n}\cdot {K_{\alpha}\over\sqrt n} ,\quad K_{\alpha}=\sqrt{\frac{2\alpha+1}{2\alpha-1}},%\label{fact}
\\
\mbox{for } \alpha=0.5\qquad     & \overline{X}_{n} \pm z_{1-x/2}\cdot  S_{n}\cdot {\sqrt {2\ln n}\over\sqrt {n}},%\label{fact1}  
\\
\mbox{for } \alpha<0.5\qquad       &  (\overline{X}_{n} - q_{x/2}\cdot  S_{n}\cdot n^{-\alpha },\ \ \overline{X}_{n} + q_{1-x/2}\cdot S_{n}\cdot n^{-\alpha }).%\label{fact2}
\end{align*}
Notably, the first of these confidence intervals differs from the classical confidence interval for the mean 
$(\overline{X}_{n} \pm z_{1-x/2}S_{n}/\sqrt n)$ 
just by a factor $K_{\alpha}$. The latter is larger than 1, as it should, in view of a positive correlation among the 
sample observations.  The correction factor $K_{\alpha}$ becomes negligible in the case of a very strong adaptation, 
$\alpha\gg1$, when the dependence due to common ancestry can be neglected. 

\begin{remark}
 Observe that our standing assumption  $\lambda=1$, see Definition \ref{def}, of 
having one speciation event per unit of time causes no loss of generality.
To incorporate an arbitrary speciation rate  $\lambda$ one has to replace in our 
formulae parameters $\alpha$ and $\sigma^2$ by $\alpha/\lambda$ and $\sigma^2/\lambda$. 
This transformation corresponds to the time scaling by factor $\lambda$ in Eq. \eqref{eqOU}, 
it changes neither the optimal value $\theta$ nor the stationary variance $\sigma^2/(2\alpha)$.
\end{remark}

\section{Sampling $m$ leaves from the Yule $n$-tree}\label{Syou}
Here we consider the Yule $n$-tree, see Definition \ref{def} and study some properties of its 
subtree joining $m$ randomly (without replacement) chosen tips, where $m\in[2, n]$. 
In particular, we compute the joint Laplace 
transform of the height of the Yule $n$-tree
$U_n=T_1+\ldots+T_n$
and $\tau^{(n)}$,
the height of the most recent common ancestor for two randomly sampled tips, see Fig. \ref{tr}. 
For other results concerning the distribution of $\tau^{(n)}$ and $U_{n}$ see also
\cite{TGer2008a,TGer2008b,AMooetal,SSagKBar2012,TSta2008,TreeSim1,TStaMSte2012,MSteAMcK2001}.

\begin{lemma}\label{MC}
Consider a random $m$-subtree of the conditioned Yule $n$-tree. It has 
\mbox{
$m-1$} bifurcating events. 
Let  $K^{(n,m)}_1< \ldots< K^{(n,m)}_{m-1}$ be the 
consecutive numbers of the bifurcation events in the Yule $n$-tree (counted from the root toward the leaves) 
corresponding to the $m-1$ bifurcating events of the $m$-subtree. Put $K^{(n,m)}_0=0$ and $K^{(n,m)}_m=n$.
The sequence $(K^{(n,m)}_m,\ldots, K^{(n,m)}_0)$ forms a time inhomogeneous Markov chain with transition probabilities

\begin{align*}
{\rm P}(K^{(n,m)}_{j-1}=k|K^{(n,m)}_j=i)=p^{(j)}_{ik},\quad  1\le j<k<i\le n,
\end{align*}
where $p^{(1)}_{i,0}=1$ for  all $i\ge 1$, and
\begin{align*}
p^{(j)}_{ik}={j(j-1)\over ik}\prod_{l=k+1}^{i-1}{(l+1-j)(l+j)\over l^2},\quad  j=2,\ldots,m.
\end{align*}
\end{lemma}
\begin{proof}
Tracing the lineages  of $m$ randomly sampled tips of the Yule $n$-tree towards the root, 
the first coalescent event can be viewed as the success in a sequence of independent Bernoulli trials. 
This argument leads to the expression cf \cite{TreeSim1}

\begin{align*}
{\rm P}(K^{(n,m)}_{m-1}=k|K^{(n,m)}_m=n)&=\left(1-{{m\choose2}\over{n\choose2}}\right)\cdots \left(1-{{m\choose2}\over{ k+2\choose2}}\right){{m\choose2}\over{ k+1\choose2}}\\
&={m(m-1)\over nk}\prod_{i=k+1}^{n-1}{(i+1-m)(i+m)\over i^2},\quad  k=m-1,\ldots,n-1,
\end{align*}
confirming the formula stated for the transition probabilities $p^{(m)}_{nk}$. The transition probabilities $p^{(j)}_{ik}$ for $j=2,\ldots,m-1$ are obtained similarly. Notice, as a check, that \mbox{
$p^{(j)}_{j,j-1}=1$}.
\end{proof}
\begin{lemma}\label{new}
Consider the inter-bifurcation times for the $m$-subtree of the Yule $n$-tree

\begin{align*}
\chi^{(n,m)}_{j}&=T_{K^{(n,m)}_{j-1}+1}+\ldots+T_{K^{(n,m)}_{j}},\quad j=1,\ldots,m,
\end{align*}
so that $U_n=\chi^{(n,m)}_{1}+\ldots+\chi^{(n,m)}_{m}$ for any $m\le n$, and $\tau^{(n)}=\chi^{(n,2)}_{2}$. 
Then for $x_j>-1$ we have

\begin{align*}
&\Expectation{\exp\Big\{-\sum\limits_{j=1}^{m}x_{j}  \chi^{(n,m)}_{j}\Big\}}   = \sum_{k_1=1}^{n-1}\sum_{k_2=k_1+1}^{n-1}\ldots\sum_{k_{m-1}=k_{m-2}+1}^{n-1}\prod_{j=1}^m
p^{(j)}_{k_j,k_{j-1}}{b_{k_j, x_{j-1}}\over b_{k_{j-1},x_{j-1}}},
\end{align*}
where $k_m=n$, $k_0=0$, and 
\[
b_{n,x}={1\over 1+x}\cdot{2\over 2+x}\cdot \ldots \cdot{n\over n+x}={\Gamma(n+1) \Gamma(x+1)\over  \Gamma(n+x+1)},\quad x>-1.
\]
\end{lemma}

\begin{proof} 
The Laplace transform of the sum of independent exponentials:

\begin{align*}
&\Expectation{\exp\{-x_1  \chi^{(n,m)}_{0}-\ldots-x_{m} \chi^{(n,m)}_{m}\}\vert(K^{(n,m)}_{m-1},\ldots K^{(n,m)}_1)=(k_{m-1},\ldots k_1)}  \\
&\qquad \qquad \qquad \qquad \qquad \qquad \qquad \qquad = \prod_{j=1}^m{k_{j-1}+1\over x_{j-1}+k_{j-1}+1}\cdots{ k_j\over x_{j-1}+k_j}
\end{align*}
together with Lemma \ref{MC} implies the stated equality

\begin{align*}
&%\Expectation{\exp\{-x_0  \tau^{(n,m)}_{0}-\ldots-x_{k-1} \tau^{(n,m)}_{m-1}\}}  
\Expectation{\exp\{-\sum\limits_{j=1}^{m}x_{j}  \chi^{(n,m)}_{j}\}}  
\\&\qquad \qquad 
= 
%\sum_{1\le k_1<\ldots<k_{m-1}<n}
\sum_{k_1=1}^{n-1}\sum_{k_2=k_1+1}^{n-1}\ldots\sum_{k_{m-1}=k_{m-2}+1}^{n-1}
\prod_{j=1}^mp^{(j)}_{k_j,k_{j-1}}{k_{j-1}+1\over x_{j-1}+k_{j-1}+1}\cdots{ k_j\over x_{j-1}+k_j}.
\end{align*}

\end{proof}
\begin{lemma}\label{jLa} The joint Laplace 
transform of the height of the Yule $n$-tree
and 
the height of the most recent common ancestor for two randomly sampled tips is given by
 \[
\Expectation{e^{-xU_n-y\tau^{(n)}}} ={2(n+1)b_{n,x+y}\over (n-1)}\sum_{k=1}^{n-1}{b_{k,x}\over (k+2)(k+1)b_{k,x+y}}.
\]
In particular, 
\be \label{eqLapT}\Expectation{e^{-xU_n}} = b_{n,x},\ee
\be\Var{e^{-xU_n}} = b_{n,2x}-b_{n,x}^2,\ee 
and, denoting the harmonic number
$h_n:=1+1/2+\ldots+1/n$,

%\be
\begin{align}\label{eqLaptau}
\Expectation{e^{-y\tau^{(n)}}}  =
\left\{
\begin{array}{ll}
{2-(n+1)(y+1)b_{n,y}\over(n-1)(y-1)},  & \mbox{ for }  y\neq 1,   \\
 \frac{2}{n-1}(h_n-1)-\frac{1}{n+1} , &   \mbox{ for } y=1. 
\end{array}
\right.
\end{align}
\end{lemma}
\begin{proof}
Turning to Lemma \ref{MC} with $m=2$ we get

\begin{align*}%\label{pi}
p^{(2)}_{n,k}&={2\over nk}\prod_{i=k+1}^{n-1}{(i-1)(i+2)\over i^2}={2(n+1)\over (n-1)(k+2)(k+1)},\quad k=1,\ldots,n-1,
\end{align*}
and according to Lemma \ref{new}

\begin{align}\label{jL}
\Expectation{e^{-x (U_n-\tau^{(n)})-y \tau^{(n)}}}  & = \sum_{k=1}^{n-1}p^{(2)}_{n,k}{1\over x+1}\cdots{k\over x+k}{k+1\over y+k+1}\cdots{ n\over y+ n}\nonumber\\
&={2(n+1)b_{n,y}\over (n-1)}\sum_{k=1}^{n-1}{b_{k,x}\over (k+2)(k+1)b_{k,y}}.
\end{align}
This implies the main formula claimed by Lemma \ref{jLa} giving $\Expectation{e^{-xU_n}} = b_{n,x}$ after putting $y=0$.
With $x=0$,

\begin{align*}
\Expectation{e^{-y \tau^{(n)}}}  &={2(n+1)b_{n,y}\over (n-1)}\sum_{k=1}^{n-1}{1\over (k+2)(k+1)b_{k,y}}\\
&={2(n+1)!\over (n-1)\Gamma(y+n+1)}\sum_{k=1}^{n-1}{\Gamma(k+1+y)\over \Gamma(k+3)}.
\end{align*}
When $y=1$ this directly becomes
\bd
\Expectation{e^{- \tau^{(n)}}}= \frac{2}{n-1}(h_n-1)-\frac{1}{n+1}.
\ed
In the case of $y\neq 1$ we use the following relation (easily verified by induction when $z\neq y$)

\be
\sum_{k=1}^{n-1}{\Gamma(k+y)\over \Gamma(k+z+1)}={ \Gamma(n+z)\Gamma(y+1)-\Gamma(z+1)\Gamma(n+y)\over \Gamma(z+1)\Gamma(n+z)(z-y)}
\label{geq}
\ee
to derive 

\begin{align*}
\Expectation{e^{-y \tau^{(n)}}}  
&={ 2\Gamma(n+1+y)-\Gamma(n+2)\Gamma(y+2)\over (n-1)\Gamma(y+n+1)(y-1)}={2-(n+1)(y+1)b_{n,y}\over(n-1)(y-1)}.
\end{align*}
\end{proof}

\begin{lemma}\label{LapT}
As $n\to\infty$ for positive $x$ and $y$ we have the following asymptotic results

\begin{align*}
\Expectation{e^{-xU_n}} &\sim  \Gamma(x+1)n^{-x},\\
\Expectation{e^{-y\tau^{(n)}}} &\sim \left\{ 
\begin{array}{ll}
\frac{1+y}{1-y}\Gamma(y+1)\cdot n^{-y}
,&\mbox{if } 0<y<1,\\
2n^{-1}\ln n
,&\mbox{if } y=1,\\
{2\over y-1}n^{-1}
,&\mbox{if } y>1,
\end{array}
\right.
%\label{Laptau}
\\
\Expectation{e^{-xU_n-y\tau^{(n)}}} &\sim \left\{ 
\begin{array}{ll}
C_{x,y} n^{-x-y}
,&\mbox{if } 0<y<1,\\
2\Gamma(x+1)n^{-x-1}\ln n
,&\mbox{if } y=1,\\
{2\Gamma(x+1)\over y-1} \cdot n^{-x-1}
,&\mbox{if } y>1,
\end{array}
\right.
%\label{ajLa}
\end{align*}
where 
\[C_{x,y}=2\Gamma(x+y+1)\sum_{k=1}^{\infty}{b_{k,x}\over (k+2)(k+1)b_{k,x+y}}.\]
 \end{lemma}
\begin{proof}
The stated results are obtained from Lemma \ref{jLa} using the first of the following three asymptotic properties of the function $b_{n,x}$

\begin{align*}
&b_{n,x}\sim  \Gamma(x+1)n^{-x}, \ n\to\infty,  %\label{b3}
\\
&{1-(n+1)b_{n,x}\over x-1}\to h_n-1, \ x\to1,%&\label{b1}
\\
&x^{-1}(1-b_{n,x})\to h_n, \ x\to0. %\label{b2}
\end{align*}
These three relations will often be used tacitly in what follows.
\end{proof}

\section{Interspecies correlation}
\label{Sout}
Denote by $\mathcal Y_n$ the $\sigma$--algebra containing all information on the Yule $n$-tree. 
The scaled trait values $Y_{i}^{(n)}:={X_i^{(n)}-\theta\over\sigma/\sqrt{2\alpha}}$, in view of Eq. \eqref{solOU}, are conditionally
normal with

\begin{align*}
\Et{Y_i^{(n)}}&=\delta e^{-\alpha U_n}, %\label{ey}
\\
\Vart{Y_i^{(n)}}&=1-e^{-2 \alpha U_n},%\label{vy}
\end{align*}
which together with the results from Section \ref{Syou} entails

\begin{align*}
\Expectation{Y_i^{(n)}}&=\delta b_{n,\alpha}, 
\\
\Var{Y_i^{(n)}}&=1-\delta^2b_{n,\alpha}^2+(\delta^2-1)b_{n,2\alpha}.
\end{align*}
\begin{lemma}\label{cy}
In the framework of the Yule-Ornstein-Uhlenbeck model, for an arbitrary pair of traits we have

\begin{align*}
\covt{Y_i^{(n)}}{Y_j^{(n)}}&=e^{-2\alpha \tau_{ij}^{(n)}}-e^{-2 \alpha U_n},
\end{align*}
where $\tau_{ij}^{(n)}$ is the backward time to the most recent common ancestor of the tips $(i,j)$.
\end{lemma}
\begin{proof}
 Denote by $ Y_{ij}^{(n)}$ the normalized trait value of the most recent common ancestor of the tips $(i,j)$. 
Let $\mathcal Y_{ij}^{(n)}$ stand for the $\sigma$--algebra generated by $(\mathcal Y_n,Y_{ij}^{(n)})$, then using Eq. \eqref{solOU} we get
\[\covariance{Y_i^{(n)}}{Y_j^{(n)}|\mathcal Y_{ij}^{(n)}}=0,\qquad \Expectation{Y_i|\mathcal Y_{ij}^{(n)}}=\Expectation{Y_j|\mathcal Y_{ij}^{(n)}}=e^{-\alpha \tau_{ij}^{(n)}}Y_{ij}^{(n)},\]
implying the statement of this lemma

\begin{align*}
\covt{Y_i^{(n)}}{Y_j^{(n)}}&=\Vart{e^{-\alpha \tau_{ij}^{(n)}}Y_{ij}^{(n)}}
\\
&=e^{-2\alpha \tau_{ij}^{(n)}}(1-e^{-2 \alpha (U_n-\tau_{ij}^{(n)})})=e^{-2\alpha \tau_{ij}^{(n)}}-e^{-2 \alpha U_n}.
\end{align*}
 \end{proof}

\begin{lemma} \label{ro}
Consider the interspecies correlation coefficient, the unconditioned correlation between two randomly sampled trait values
 \[ \rho_{n} ={1\over n(n-1)}\sum_{i}\sum_{j\ne j}{\covariance{X_i^{(n)}}{X_j^{(n)}}\over\Var{X_1^{(n)}}}={1\over n(n-1)}\sum_{i}\sum_{j\ne j}{\covariance{Y_i^{(n)}}{Y_j^{(n)}}\over\Var{Y_1^{(n)}}}. \]
If $\alpha \neq 0.5$, then
\[
\rho_{n} = 1-\frac{2\alpha(n-1)+(n+1)((1+2\alpha)b_{n,2\alpha}-1)}{(n-1)(2\alpha-1)(1+(\delta^2-1)b_{n,2\alpha} -\delta^2b_{n,\alpha}^2 )},
\]
and in the case of $\alpha=0.5$
\[
\rho_{n} = 1-\frac{n+1}{n-1}\frac{n+2-2h_n}{n+\delta^{2}(1-(n+1)b_{n,0.5}^{2})}.
\]
\end{lemma}
\begin{proof}
 According to Lemma \ref{cy} we have,

\begin{align*}
{2\over n(n-1)}\sum_{i<j}\covariance{Y_i^{(n)}}{Y_j^{(n)}}&=\Expectation{e^{-2\alpha \tau^{(n)}}-e^{-2 \alpha U_n}} + \delta^{2}\Var{e^{-\alpha U_n}}%\label{cov}
\end{align*}
leading to 
\[
 \rho_{n} = 1-\frac{1-\Expectation{e^{-2\alpha \tau^{(n)}}} }{1-\Expectation{e^{-2\alpha U_n}}+\delta^{2}\Var{e^{-\alpha U_n}}}.%\label{ron}
\]
Applying the results of Section \ref{Syou} we arrive at the asserted relations for $\rho_{n}$. Observe that asymptotically as $n\to\infty$ the interspecies correlation coefficient decays to $0$ as
\[
\rho_{n} \sim \left\{
\begin{array}{ll}
%c_{\alpha,\delta}
\frac{2}{1-2\alpha}\Gamma(1+2\alpha)+\delta^{2}\left(\Gamma(1+2\alpha)-\Gamma^{2}(\alpha+1)\right)n^{-2\alpha} & 0<\alpha<0.5,  \\
2n^{-1}\ln n& \alpha=0.5, \\
\frac{2}{2\alpha-1}n^{-1}& \alpha>0.5.
\end{array}
\right.
\]
\end{proof}
\begin{lemma}\label{D}
Consider the sample mean $\overline{Y}_{n}=n^{-1}(Y_1^{(n)}+\ldots+Y_n^{(n)})$ and the sample variance 
\[
D_{n}^2=\frac{1}{n-1}\sum\limits_{i=1}^{n} (Y_{i}^{(n)}-\overline{Y}_{n})^{2},
\]
of the scaled trait values. For all $\alpha >0$ we have $\Expectation{\overline{Y}_{n}}=\delta b_{n,\alpha}$. For  $\alpha \neq 0.5$

\begin{align*}
\Var{\overline{Y}_{n}}&={1+2\alpha-(4\alpha n+1+2\alpha)b_{n,2\alpha}\over(2\alpha-1)n}+\delta^2(b_{n,2\alpha}-b_{n,\alpha}^2),\\
\Expectation{D^{2}_{n}}&= 1+\frac{(1+2\alpha)(n+1)b_{n,2\alpha}-2}{(2\alpha-1)(n-1)},
\end{align*}
and in the singular case $\alpha=0.5$

\begin{align*}
\Var{\overline{Y}_{n}}&=\frac{2(h_n-1)}{n}+\frac{\delta^2-1}{n+1}-\delta^2b_{n,0.5}^2,\\
\Expectation{D^{2}_{n}}&= \frac{n-2h_n}{n-1}.
\end{align*}
\end{lemma}

\begin{proof}
 Obviously, $\Expectation{\overline{Y}_{n}}=\Expectation{Y_i^{(n)}}=\delta b_{n,\alpha}$.
To prove the other assertions we turn to  \cite{SSagKBar2012}, where the concept of interspecies correlation was originally introduced. It was shown there that the variance of the sample average and the expectation
of the sample variance can be compactly expressed as

\begin{align*}
\Var{\overline{X}_{n}}&=\Big(\frac{1}{n}+\frac{n-1}{n}\rho_{n}\Big)\Var{X^{(n)}_1},\\ 
\Expectation{S^{2}_{n}} &=(1-\rho_{n})\Var{X^{(n)}_1}.
\end{align*}
Since $\Var{Y^{(n)}_1}={2\alpha\over\sigma^2}\Var{X^{(n)}_1}$, $\Var{\overline{Y}_{n}}={2\alpha\over\sigma^2}\Var{\overline{X}_{n}}$, and $\Expectation{D^{2}_{n}}={2\alpha\over\sigma^2}\Expectation{S^{2}_{n}}$ 
it remains to combine Lemma \ref{ro} with the known expression for $\Var{Y^{(n)}_1}$.

A more direct proof of Lemma \ref{D} can be obtained using the following result on conditional expectations.
\end{proof} 
\begin{lemma}\label{barY}
We have

\begin{align*}
\Et{\overline{Y}_{n}}  &=\delta e^{-\alpha U_{n}}, \\
\Et{\overline{Y}_{n}^2}  &=  n^{-1}+(1-n^{-1})\Et{e^{-2\alpha \tau^{(n)}}} - e^{-2\alpha U_{n}}+\delta^2e^{-2\alpha U_{n}},\\
\Vart{\overline{Y}_{n}^2}  &=  n^{-1}+(1-n^{-1})\Et{e^{-2\alpha \tau^{(n)}}} - e^{-2\alpha U_{n}}.
\end{align*}
\end{lemma}
\begin{proof} The main assertion follows from

\begin{align*}
\Vart{Y_1^{(n)}+\ldots+Y_n^{(n)}}&=n(1-e^{-2 \alpha U_n})+2\sum_{i<j}(e^{-2\alpha \tau_{ij}^{(n)}}-e^{-2 \alpha U_n})\\
&=n-n^2e^{-2 \alpha U_n}
+n(n-1)\Et{e^{-2\alpha \tau^{(n)}}}.
\end{align*}

\end{proof}

\section{Proof of Theorem \ref{tMart}}\label{SecMart}
\begin{lemma}\label{Vma}
Put
$V_n^{(x)}:=b_{n,x}^{-1}\cdot e^{-xU_n}$ with $\Expectation{V_n^{(x)}}=1$. 
For any $x>-1$ the sequence $\{V_n^{(x)},\mathcal Y_n\}_{n\ge0}$ forms a martingale 
converging a.s. and in $L^2$. 
Moreover, $(U_n-\log n)$ converges in distribution to a random variable having the 
standard Gumbel distribution.
\end{lemma}
\begin{proof}
The martingale property is obvious 

$$\Et{V_{n+1}^{(x)}}=b_{n+1,x}^{-1}\cdot e^{-xU_n}\Expectation{e^{-xT_{n+1}}}=b_{n,x}^{-1}\cdot e^{-xU_n}=V_n^{(x)}.$$
Since the second moments
\bd
\Expectation{(V_n^{(x)})^2}=b_{n,x}^{-2}\cdot \Expectation{e^{-2xU_n}}=\frac{b_{n,2x}}{(b_{n,x})^2}
\ed
are uniformly bounded over $n$, we may conclude that $V_n^{(x)}\to V^{(x)}$ a.s. and in $L^2$
with $\Expectation{V^{(x)}}=1$. 
It follows that $\Expectation{V_n^{(x)}}\to 1$, and therefore, $\Expectation{e^{-x(U_n-\log n)}}\to\Gamma(x+1)$. The latter is a convergence of Laplace transforms confirming the stated 
convergence in distribution.
\end{proof}
Observe that the Gumbel limit for $U_n-\log n$ can be obtained using the 
classical extreme value theory, in view of the representation 
\[U_n\stackrel{d}{=}\sum_{i=1}^ni^{-1}E_i\stackrel{\mathcal{D}}{=}\max(E_1,\ldots, E_n)\]
in terms of independent exponentials with parameter 1. 
Notice also that $U_{n+1}/2$ has the same distribution as the total branch length of Kingman's $n$-coalescent.
\begin{lemma}\label{Hma}
Denote by $\mathcal{F}_{n}$ the $\sigma$--algebra containing information on the Yule $n$-tree realization 
as well as the corresponding information on the evolution of trait values. Set

\begin{align*}
H_n:&=(n+1)e^{(\alpha-1)U_n}\overline Y_n,\quad n\ge0.
\end{align*}
The sequence $\{H_{n},\mathcal{F}_{n}\}_{n\ge0}$ forms a martingale with $\Expectation{H_{n}}=H_{0}=\delta$.
\end{lemma}
\begin{proof}
Notice that, 

\begin{align*}
\Expectation{e^{(\alpha-1)T_{n+1}}\sum_{i=1}^{n+1}Y_{i}^{(n+1)} \vert \mathcal{F}_{n}}&=
\Expectation{e^{- T_{n+1}}}\left(\sum_{i=1}^{n}Y_{i}^{(n)}+n^{-1}\sum_{j=1}^{n}Y_{j}^{(n)}\right)\\
&={n+1\over n+2}{n+1\over n}\sum_{i=1}^{n}Y_i^{(n)}=\frac{(n+1)^{2}}{n+2}\overline{Y}_{n}.
\end{align*}
Hence

\begin{align*}
\Expectation{H_{n+1}|\mathcal F_n}&={n+2\over n+1}e^{(\alpha-1)U_n}\Expectation{e^{(\alpha-1)T_{n+1}}\sum_{i=1}^{n+1}Y_i^{(n+1)}|\mathcal F_n}=H_n.
\end{align*}
\end{proof}
\begin{lemma}\label{byt}
For all positive $\alpha$ we have
 $\Var{\Et{e^{-2\alpha \tau^{(n)}}}}= O(n^{-3})$ as $n\to\infty$. 
\end{lemma}
\begin{proof} 
For a given realization of the Yule $n$-tree we  denote by $\tau^{(n)}_{1}$ and  $\tau^{(n)}_{2}$ two independent versions of $\tau^{(n)}$ corresponding to two independent choices of pairs of tips out of $n$ available. We have,
%Then, using distribution \eqref{pi} we get

\begin{align*}
&\Expectation{\Big(\Et{e^{-2\alpha \tau^{(n)}}}\Big)^2}=\Expectation{\Et{e^{-2\alpha (\tau_1^{(n)}+\tau_2^{(n)})}}}=\Expectation{e^{-2\alpha (\tau_1^{(n)}+\tau_2^{(n)})}}.
\end{align*}
Writing

$$\pi_{n,k}:=p^{(2)}_{n,k},\qquad f(a,k,n)={k+1\over a+k+1}\cdots{n\over a+n}$$ 
and using the ideas of Section \ref{Syou} we obtain

\begin{align*}
\Expectation{\Big(\Et{e^{-2\alpha \tau^{(n)}}}\Big)^2}&=\sum_{k=1}^{n-1}f_{4\alpha}(k,n)\pi_{n,k}^2\\
&\quad +2\sum_{k_1=1}^{n-1}\sum_{k_2=k_1+1}^{n-1}f_{2\alpha}(k_1,k_2)f_{4\alpha}(k_2,n)\pi_{n,k_1}\pi_{n,k_2}.
\end{align*}
On the other hand,

\begin{align*}
\Big(\Expectation{e^{-2\alpha\tau^{(n)}}} \Big)^2&=\Big(\sum_{k_1}f_{2\alpha}(k_1,n)\pi_{n,k_1}\Big)\Big(\sum_{k_2}f_{2\alpha}(k_2,n)\pi_{n,k_2}\Big).
\end{align*}
Taking the difference between the last two expressions we find 

\begin{align*}
 \Var{\Et{e^{-2\alpha \tau^{(n)}}}}&=\sum_k\Big(f_{4\alpha}(k,n)-f_{2\alpha}(k,n)^2\Big)\pi_{n,k}^2\\
&\hspace{-1cm} +2\sum_{k_1=1}^{n-1}\sum_{k_2=k_1+1}^{n-1}f_{2\alpha}(k_1,k_2)
\Big(f_{4\alpha}(k_2,n)-f_{2\alpha}(k_2,n)^2\Big)\pi_{n,k_1}\pi_{n,k_2}.
\end{align*}
Using the simple equality
\[a_{1}\cdots a_{n}-b_{1}\cdots b_{n}=\sum_{i=1}^{n}b_{1}\cdots b_{i-1}(a_{i}-b_{i})a_{i+1}\cdots a_{n}\]
we see that it suffices to prove that,

\begin{align*}
&\sum_{k=1}^{n-1} A_{n,k}\pi_{n,k}^2= O(n^{-4}),%\label{pi1}
\\
&\sum_{k_1=1}^{n-1}\sum_{k_2=k_1+1}^{n-1}f_{2\alpha}(k_1,k_2)A_{n,k_2}\pi_{n,k_1}\pi_{n,k_2}=O(n^{-3}),%\label{pi2}
\end{align*}
where
\[A_{n,k}:=\sum_{j=k+1}^{n-1}f_{2\alpha}(k,j)^2 \Big({2\alpha\over 2\alpha+j+1}\Big)^2f_{4\alpha}(j,n).\]
To verify these two asymptotic relations  observe that
\[A_{n,k}<{k+1\over4\alpha+k+1}\cdots{n\over4\alpha+n}\sum_{i=k+1}^n{4\alpha^2\over (2\alpha+i)^2}<{4\alpha^2b_{n,4\alpha}\over b_{k,4\alpha}} \sum_{i=k+1}^n {1\over i(i-1)}<{4\alpha^2b_{n,4\alpha}\over k b_{k,4\alpha}}.\]
Since
$\pi_{n,k} = \frac{2(n+1)}{(n-1)(k+2)(k+1)}$,
it follows

\begin{align*}
\sum_{k=1}^{n-1} A_{n,k}\pi_{n,k}^2<c_1b_{n,4\alpha} \sum_{k=1}^{n-1} {1\over k^5 b_{k,4\alpha}}<c_2n^{-4\alpha} \sum_{k=1}^{n} n^{4\alpha-5}<c_2n^{-4}, 
\end{align*}
and

\begin{align*}
\sum_{k_1=1}^{n-1}\sum_{k_2=k_1+1}^{n-1}f_{2\alpha}(k_1,k_2)A_{n,k_2}\pi_{n,k_1}\pi_{n,k_2}&<c_3b_{n,4\alpha} \sum_{k_1=1}^{n-1}\sum_{k_2=k_1+1}^{n-1} {b_{k_2,2\alpha}\over b_{k_1,2\alpha} b_{k_2,4\alpha}} {1\over k_1^2 k_2^3}\\
&\hspace{-2cm}<c_4n^{-4\alpha} \sum_{k_2=2}^{n} k_2^{2\alpha-3} \sum_{k_1=1}^{k_2}k_1^{2\alpha-2}<c_4n^{-4\alpha} \sum_{k_2=2}^{n} k_2^{4\alpha-4} <c_4n^{-3}.
\end{align*}
\end{proof}

{\sc Proof of Theorem \ref{tMart} (i) and (ii).}  
Let $\alpha>0.5$. To establish the stated normal approximation it is enough to prove the convergence in probability of the first two conditional moments 

\begin{align*}
(\mu_n,\sigma^2_n)&:= \big(\sqrt n \Et{\overline Y_n},\ n \Vart{\overline Y_n}\big)\stackrel{P}{\to}\Big(0, {2\alpha+1\over 2\alpha-1}\Big),\quad n\to\infty,
\end{align*}
since then, due to the conditional normality of $\overline Y_n$, we will get the following convergence of characteristic functions
\[\Expectation{e^{i \gamma \sqrt n \cdot\overline Y_{n}}}=\Expectation{e^{i\mu_n\gamma-\sigma_n^2\gamma^2/2}}\to e^{-{2\alpha+1\over 2(2\alpha-1)}\gamma^2}.\]
Now, due to Lemma \ref{barY} we can write

\begin{align*}
(\mu_n,\sigma^2_n)&=  \big(\sqrt n \delta e^{-\alpha U_{n}},\ 1+(n-1)\Et{e^{-2\alpha \tau^{(n)}}} - ne^{-2\alpha U_{n}}\big).
\end{align*}
Using relations from Section \ref{Sout} we see that 

$$\Expectation{\sigma^2_n}=1-nb_{n,2\alpha}+{2-(n+1)(2\alpha+1)b_{n,2\alpha}\over 2\alpha-1}\to{2\alpha+1\over 2\alpha-1}.$$ 
It remains to observe that on one hand, according to Lemma \ref{byt}

$$1+(n-1)\Et{e^{-2\alpha \tau^{(n)}}}\stackrel{P}{\to}{2\alpha+1\over 2\alpha-1},$$ 
and on the other hand, $ne^{-2\alpha U_{n}}\stackrel{P}{\to}0$, implying that $\sigma^2_n\stackrel{P}{\to}{2\alpha+1\over 2\alpha-1}$. This together with $\mu_n\to0$ holding in $L^2$ and therefore in probability, entails $(\mu_n,\sigma^2_n)\stackrel{P}{\to}(0,{2\alpha+1\over 2\alpha-1})$, finishing the proof of part (i).
Part (ii) is proven similarly.\\

{\sc Proof of Theorem \ref{tMart} (iii).}  Let  $0<\alpha<0.5$. Turning to Lemma \ref{Hma} observe that the martingale $H_{n}=(n+1)e^{(\alpha-1)U_n}\overline Y_n$ has uniformly bounded second moments. Indeed, due to Lemma \ref{barY}

\begin{align*}
\Expectation{H_{n}^{2}} &=(n+1)^2\Expectation{e^{2(\alpha-1)U_n}\Et{\overline Y_n^2}}\\
& < c_1n \Expectation{e^{-2(1-\alpha)U_{n}}}+c_2n^2\Expectation{e^{-2(1-\alpha)U_{n}-2\alpha \tau^{(n)}}}
+c_3n^2\Expectation{e^{-2\alpha  U_{n}}}.
\end{align*}
Thus, according to Lemma \ref{LapT} we have $\sup\limits_{n} \Expectation{H_{n}^{2}} <\infty$. Referring to the martingale $L^{2}$-convergence theorem we conclude that $H_{n}\to H_\infty$ almost surely and in $L^{2}$. Due to Lemma  \ref{Vma} it follows that 

$$n^\alpha \overline Y_n={n^\alpha b_{n,\alpha-1}\over n+1}V_n^{(\alpha-1)}H_n\to V^{(\alpha-1)}H_\infty=:Y_{\alpha,\delta}\quad\mbox{a.s. and in } L^{2}.$$
Finally, as $n\to\infty$

\begin{align*}
n^\alpha \Expectation{\overline Y_n}&=\delta n^\alpha b_{n,\alpha}\to\delta\Gamma(1+\alpha),\\
n^{2\alpha}\Expectation{\overline{Y}_{n}^{2}} & = n^{2\alpha-1}+n^{2\alpha}(1-n^{-1})\Expectation{e^{-2\alpha \tau^{(n)}}} +n^{2\alpha} (\delta^{2}-1) \Expectation{e^{-2\alpha U_{n}}}\\
&\to 
\left(\delta^{2}+\frac{4\alpha}{1-2\alpha} \right)\Gamma(1+2\alpha).
%\frac{1+2\alpha}{1-2\alpha}+(\delta^{2}-1)\Gamma(1+2\alpha).
\end{align*}

\section{Consistency of the sample variance}\label{Scons}
Recall that $\Expectation{S^{2}_{n}}={\sigma^2\over2\alpha}\Expectation{D^{2}_{n}}$, and according to Lemma \ref{D} we have $\Expectation{D_{n}^{2}}\to 1$. The aim of this section is to show that  $\Var{D_{n}^{2}}\to 0$ as $n\to\infty$ which is equivalent to
\be\label{cons}
\Expectation{D_{n}^{4}}\to1,\quad n\to\infty.
\ee
To this end we will need the following formula, see Eq. (13) in \cite{GBohAGol1969}
valid for any normally distributed vector $(Z_1, Z_2, Z_3,Z_4)$ with means $(m_1,m_2,m_3,m_4)$ and covariances
$\covariance{Z_i}{Z_j}=c_{ij}$:
\[
\covariance{Z_1Z_2}{Z_3Z_4} = m_1m_3c_{24}+m_1m_4c_{23}+m_2m_3c_{14}+m_2m_4c_{13}+c_{13}c_{24} + c_{14}c_{23}.
\]
In the special case with $m_i=m$ it follows

\begin{align}
\label{cprod}
\Expectation{Z_1Z_2Z_3Z_4} &= m^4+m^2(c_{12}+c_{13}+c_{14}+c_{23}+c_{24}+c_{34})+c_{12}c_{34}+c_{13}c_{24} + c_{14}c_{23}.
\end{align}

Writing $Y_i$ instead of $Y_i^{(n)}$ we use the representation

\begin{align*}
D_{n}^2&=\frac{n}{n-1}\left(\frac{1}{n}\sum\limits_{i=1}^{n} Y_{i}^{2}-\overline{Y}_{n}^2\right)
=\frac{1}{n}\sum\limits_{i}Y_{i}^{2}-\frac{2}{n(n-1)}\sum_i\sum_{j>i}Y_{i}Y_{j}
\end{align*}
to find out that

\begin{align*}
\Expectation{D_{n}^{4}}  &=\frac{1}{n^2} \Big(\sum\limits_{i}\Expectation{Y_i^{4}}+2\sum_i\sum_{j>i}\Expectation{Y_i^{2}Y_j^2}\Big)\\
&\quad - \frac{4}{n^2(n-1)} \Big(\sum_i\sum_{j>i}\Expectation{Y_{i}^{3}Y_{j}}+\sum_i\sum_{j>i}\Expectation{Y_{i}Y_{j}^{3}}+\sum_i\sum_{j>i}\sum_{k\neq i,j}\Expectation{Y_{i}^{2}Y_{j}Y_k}\Big)\\
&\quad + \frac{4}{n^2(n-1)^2}\Big(\sum_i\sum_{j>i}\Expectation{Y_{i}^{2}Y_{j}^2}+\sum_i\sum_{j>i}\sum_{k\neq i,j}\Expectation{Y_{i}^{2}Y_{j}Y_k}+\sum_i\sum_{j>i}\sum_{k\neq i,j}\Expectation{Y_{i}Y_{j}^{2}Y_k}\\
&\quad\qquad\qquad\qquad\qquad\qquad\quad +\sum_i\sum_{j>i}\ \sum_{k\neq i,j}\ \sum_{m>k;\ m\neq i,j}\Expectation{Y_{i}Y_{j}Y_kY_m}\Big).
\end{align*}
Denoting by $(W_{1},W_{2},W_{3},W_{4})$ a random sample without replacement of four trait values out of $n$ available, so that

\begin{align*}
\Expectation{W_1^{4}}  &=n^{-1} \sum\limits_{i}\Expectation{Y_i^{4}},\\
\Expectation{W_{1}^{3}W_{2}}&= \frac{1}{n(n-1)} \sum_{i}\sum_{j\neq i}\Expectation{Y_{i}^{3}Y_{j}},\\
\Expectation{W_{1}^{2}W_{2}^2}&= \frac{1}{n(n-1)}\sum_{i}\sum_{j\neq i}\Expectation{Y_{i}^{2}Y_{j}^2},\\
\Expectation{W_{1}^{2}W_{2}W_{3}}&= \frac{1}{n(n-1)(n-2)}\sum_{i}\sum_{j\neq i}\,\sum_{k\neq i,j}\Expectation{Y_{i}^{2}Y_{j}Y_k},
\\
\Expectation{W_{1}W_{2}W_{3}W_{4}}&= \frac{1}{n(n-1)(n-2)(n-3)}\sum_{i}\sum_{j\neq i}\,\sum_{k\neq i,j}\ \sum_{m\neq i,j,k}\Expectation{Y_{i}Y_{j}Y_kY_m},
\end{align*}
we derive

\begin{align}
\Expectation{D_{n}^{4}}  &=n^{-1} \Expectation{W_1^{4}}-4n^{-1}\Expectation{W_{1}^{3}W_{2}}+\frac{n^2-2n+3}{n(n-1)}\Expectation{W_{1}^{2}W_{2}^2}\nonumber\\
&\quad  - \frac{2(n-2)(n-3)}{n(n-1)}\Expectation{W_{1}^{2}W_{2}W_{3}}+ \frac{(n-2)(n-3)}{n(n-1)}\Expectation{W_{1}W_{2}W_{3}W_{4}}.\label{stist}
\end{align}

We compute the five fourth-order moments in the last expression using the conditional normality of the random quadruple 
$(W_{1},W_{2},W_{3},W_{4})$ with conditional moments given by

\begin{align*}
\Et{W_i}&= \delta e^{-\alpha U_{n}},\\
\Et{W_i^2}&= 1+(\delta^2-1)e^{-2\alpha U_{n}},\\
\Vart{W_{i}}&=1-e^{-2\alpha U_{n}},\\
\covt{W_i}{W_j}&=\Et{e^{-2\alpha\tau^{(n,4)}_{ij}}}-e^{-2 \alpha U_n},\quad i,j\in\{1,2,3,4\},\quad i\ne j,
\end{align*}
where $ \tau^{(n,m)}_{ij}$ is the time to the most recent ancestor for the pair of tips $(i,j)$ among $m$ randomly chosen tips of the Yule $n$-tree. Clearly, all $ \tau^{(n,4)}_{ij}$ have the same distribution as $\tau^{(n)}$, and for 
\[\nu^{(n)}_{ij}:=\Et{e^{-2\alpha\tau^{(n,4)}_{ij}}}\]
we can find the asymptotics of 

$$\Expectation{\nu^{(n)}_{ij}}=\Expectation{e^{-2\alpha\tau^{(n)}}},\quad \Expectation{\nu^{(n)}_{ij}e^{-2\alpha U_{n}}}=\Expectation{e^{-2\alpha U_{n}-2\alpha\tau^{(n)}}}$$
using Lemma \ref{LapT}. Notice also that

$$\Expectation{(\nu^{(n)}_{ij})^2}\sim\Big(\Expectation{e^{-2\alpha\tau^{(n)}}}\Big)^2,\quad n\to\infty.$$
This follows from Lemma \ref{byt} and  Lemma \ref{LapT} as

$$\Expectation{(\nu^{(n)}_{ij})^2}=\Expectation{\Big(\Et{e^{-2\alpha\tau^{(n)}}}\Big)^2}=\Var{\Et{e^{-2\alpha\tau^{(n)}}}}+\Big(\Expectation{e^{-2\alpha\tau^{(n)}}}\Big)^2.$$

\noindent (i) With $Z_1=Z_2=Z_3=Z_4=W_1$ in  Eq. \eqref{cprod}, we obtain

\begin{align*}
\Et{W_1^{4}}& =\delta^4e^{-4\alpha U_{n}}+6\delta^2e^{-2\alpha U_{n}}(1-e^{-2\alpha U_{n}})+3(1-e^{-2\alpha U_{n}})^2,
\end{align*}
and therefore  $\Expectation{W_{1}^4}\to3$  as $n\to\infty$.\\

\noindent(ii) Using Eq. \eqref{cprod} with $Z_1=Z_2=Z_3=W_1$ and $Z_4=W_2$
we obtain

\begin{align*}
\Et{W_{1}^3W_{2}}&=\delta^4e^{-4\alpha U_{n}}
+3\delta^2e^{-2\alpha U_{n}}(1- e^{-2\alpha U_{n}})
\\
&\quad+3\delta^2e^{-2\alpha U_{n}}(\nu_{12}^{(n)}-e^{-2\alpha U_{n}})+3(1- e^{-2\alpha U_{n}})(\nu_{12}^{(n)}-e^{-2\alpha U_{n}})\\
&=3\nu_{12}^{(n)}-3(\delta^{2}-1)(1-\nu_{12}^{(n)})e^{-2\alpha U_{n}}+(\delta^{4}-6\delta^{2}+3)e^{-4\alpha U_{n}},
\end{align*}
%implying for $\alpha<0.5$
resulting in $\Expectation{W_{1}^3W_{2}}\to0$  as $n\to\infty$.\\

\noindent(iii) Eq. \eqref{cprod} with $Z_1=Z_2=W_1$ and $Z_3=Z_4=W_2$ gives

\begin{align*}
\Et{W_{1}^{2}W_{2}^{2}}&=\delta^4e^{-4\alpha U_{n}}+2\delta^2e^{-2\alpha U_{n}}(1- e^{-2\alpha U_{n}})
\\
&\quad+4\delta^2e^{-2\alpha U_{n}}(\nu_{12}^{(n)}-e^{-2\alpha U_{n}})+(1- e^{-2\alpha U_{n}})^2+2(\nu_{12}^{(n)}-e^{-2\alpha U_{n}})^2\\
&=1+2(\delta^{2}-1)e^{-2\alpha U_{n}}
\\&\quad +(\delta^{4}-6\delta^{2}+5)e^{-4\alpha U_{n}}
+4(\delta^{2}-1)\nu_{12}^{(n)}e^{-2\alpha U_{n}}+2(\nu_{12}^{(n)})^{2},
\end{align*}
so that $\Expectation{W_{1}^{2}W_{2}^{2}}\to1$  as $n\to\infty$.\\

\noindent(iv) Using a consequence of Eq. \eqref{cprod},

\begin{align*}
\Expectation{Z_{1}^{2}Z_{2}Z_{3}} &= m^{4}+m^{2}(c_{11}+2c_{12}+2c_{13}+c_{23})+ c_{11}c_{23}+2c_{12}c_{13},
\end{align*}
we get

\begin{align*}
\Et{W_{1}^2W_{2}W_3}&=\delta^4e^{-4\alpha U_{n}}
+\delta^2e^{-2\alpha U_{n}}(1- e^{-2\alpha U_{n}})
\\
&\quad+\delta^2e^{-2\alpha U_{n}}(2\nu_{12}^{(n)}+2\nu_{13}^{(n)}+\nu_{23}^{(n)}-5e^{-2\alpha U_{n}})\\
&\quad+(1- e^{-2\alpha U_{n}})(\nu_{23}^{(n)}-e^{-2\alpha U_{n}})+2(\nu_{12}^{(n)}-e^{-2\alpha U_{n}})(\nu_{13}^{(n)}-e^{-2\alpha U_{n}})\\
&=(\delta^{2}-1)e^{-2\alpha U_{n}}+(\delta^{4}-6\delta^{2}+3)e^{-4\alpha U_{n}}+2\nu_{12}^{(n)}\nu_{13}^{(n)}\\
&\quad+2(\delta^{2}-1)e^{-2\alpha U_{n}}(\nu_{12}^{(n)}+\nu_{13}^{(n)})+(1+(\delta^{2}-1)e^{-2\alpha U_{n}})\nu_{23}^{(n)}.
\end{align*}
Using the Cauchy-Schwarz inequality

$$0\le\Expectation{\nu_{12}^{(n)}\nu_{13}^{(n)}}\le \Big(\Expectation{\nu_{12}^{(n)}}\Big)^2=\Big(\Expectation{e^{-2\alpha\tau^{(n)}}}\Big)^2=o(1),$$
we obtain $\Expectation{W_{1}^{2}W_{2}W_{3}}\to0$  as $n\to\infty$.\\

\noindent (v) According to Eq. \eqref{cprod} we have

\begin{align*}
\Et{W_{1}W_{2}W_3W_4}&=\delta^4e^{-4\alpha U_{n}}\\
&\quad
+\delta^2e^{-2\alpha U_{n}}(\nu_{12}^{(n)}+\nu_{13}^{(n)}+\nu_{14}^{(n)}+\nu_{23}^{(n)}+\nu_{24}^{(n)}+\nu_{34}^{(n)}- 6e^{-2\alpha U_{n}})
\\
&\quad+(\nu_{12}^{(n)}-e^{-2\alpha U_{n}})(\nu_{34}^{(n)}-e^{-2\alpha U_{n}})\\
&\quad+(\nu_{13}^{(n)}-e^{-2\alpha U_{n}})(\nu_{24}^{(n)}-e^{-2\alpha U_{n}})\\
&\quad+(\nu_{14}^{(n)}-e^{-2\alpha U_{n}})(\nu_{23}^{(n)}-e^{-2\alpha U_{n}}),
\end{align*}
implying

\begin{align*}
\Et{W_{1}W_{2}W_3W_4}&=(\delta^{4}-6\delta^{2}+3)e^{-4\alpha U_{n}}+\nu_{12}^{(n)}\nu_{34}^{(n)}+\nu_{12}^{(n)}\nu_{34}^{(n)}+\nu_{12}^{(n)}\nu_{34}^{(n)}\\
&\quad
+(\delta^2-1)e^{-2\alpha U_{n}}(\nu_{12}^{(n)}+\nu_{13}^{(n)}+\nu_{14}^{(n)}+\nu_{23}^{(n)}+\nu_{24}^{(n)}+\nu_{34}^{(n)}).
\end{align*}
Using an estimate for $\Expectation{\nu_{12}^{(n)}\nu_{34}^{(n)}}=\Expectation{\nu_{12}^{(n)}\nu_{34}^{(n)}}=\Expectation{\nu_{12}^{(n)}\nu_{34}^{(n)}}$  similar to that we used in (iv), we find  $\Expectation{W_{1}W_{2}W_{3}W_4}\to0$  as $n\to\infty$. \\

Finally, putting the above results (i) - (v) into Eq. \eqref{stist} we arrive at Eq. \eqref{cons}.

\section*{Acknowledgments}
We are grateful to Thomas F. Hansen and Anna Stokowska for helpful suggestions and comments.  Special thanks to an anonymous referee for the constructive suggestions to the earlier version of the paper, in particular, for the remark added after Lemma \ref{Vma}.

The research of Serik Sagitov was supported by the Swedish Research Council grant 621-2010-5623. 
Krzysztof Bartoszek was supported by the Centre for Theoretical Biology at the University of Gothenburg, 
Svenska Institutets \"Ostersj\"osamarbete
scholarship nr. 11142/2013, 
Stiftelsen f\"or Vetenskaplig Forskning och Utbildning i Matematik
(Foundation for Scientific Research and Education in Mathematics), 
Knut and Alice Wallenbergs travel fund, Paul and Marie Berghaus fund, the Royal Swedish Academy of Sciences,
and Wilhelm and Martina Lundgrens research fund.

\appendix
\section{ 
All moments of $U_{n}$ and $\tau^{(n)}$}\label{appSaM}
Eq. \eqref{eqLapT} for the Laplace transforms of the random
variable $U_{n}$ can be used to calculate the moments of $U_{n}$ using,

\be
\Expectation{U_{n}^{m}}=(-1)^{m}(\partial^{m}\Expectation{e^{-xU_{n}}}/\partial x^{m})\vert_{x=0}.
\ee
For a fixed $n$ we introduce the following notation,
\begin{align*}
A(x)&=\frac{1}{x+1}\cdot \ldots \cdot \frac{1}{x+n},\\
b_{m}(x)&=\frac{1}{(x+1)^{m}}+ \ldots + \frac{1}{(x+n)^{m}},\\ 
\mathbf b_{m}(x)&=(b_{1}(x),\ldots,b_{m}(x)).
\end{align*}
Notice that  
$A(0)=1/n!$ and $b_{m}(0)=H_{n,m} $ is the $n$--th generalized harmonic number of order $m$,

\be
H_{n,m} = \sum\limits_{i=1}^{n}\frac{1}{i^{m}}.
\ee

We can write Eq. \eqref{eqLapT} as $\Expectation{e^{-x U_{n}}}=n!A(x)$. Its first derivative with respect
to $x$ is $-n!A(x)b_{1}(x)$, and the second derivative is $n!A(x)(b_{1}(x)^{2}+b_{2}(x))$.
For the general recursive formula we introduce the following notation.
We will denote by 
\mbox{
$\mathbf k=(k_{1},k_{2},\ldots)$} infinite dimensional vectors with integer--valued components, and write $\mathbf k\in\mathcal{A}_{m}$ if all $k_i\ge0$ and $|\mathbf k|:=\sum_{i=1}^{m} k_{i}i=m$.
Therefore $\mathcal{A}_{m}$ represents the set of all possible ways to represent $m$ as a sum of positive integers. We will also
use the multi--index notation  \mbox{
$\mathbf b_m(x)^{\mathbf k} =b_{1}(x)^{k_{1}}\cdot \ldots \cdot b_{m}(x)^{k_{m}}$.} 

Since $A'(x)=-A(x)b_{1}(x)$, 
and $b'_{m}(x)=-mb_{m+1}(x)$, we can show by induction that,

\begin{align}
\frac{\partial^{m}}{\partial x^{m}}\Expectation{e^{-x U_{n}}} & = 
(-1)^{m}n!A(x)\sum\limits_{\substack {\mathbf k\in \mathcal{A}_{m}}}c_{\mathbf k}
\mathbf b_m(x)^{\mathbf k},\label{moT}
\end{align}
where coefficients $c_{\mathbf k}$ are defined for all vectors $\mathbf k=(k_1,k_2,\ldots)$ with integer--valued components using the recursion,

\be\label{recc}
c_{\mathbf k}=\sum\limits_{\substack {j=0}}^m(jk_{j}+1)c_{\mathbf k,j},\ee
with $m=|\mathbf k|$ and

\begin{align*}
c_{\mathbf k,0}&=c_{(k_{1}-1,k_{2},k_3,\ldots)},\\
c_{\mathbf k,j}&=c_{(k_{1},\ldots,k_{j}+1,k_{j+1}-1,\ldots)}, \ j\ge1.
\end{align*}
The boundary conditions for the recursion of Eq. \eqref{recc} consist of two parts:
\begin{itemize}
\item $c_{\mathbf k}=0$, if all $k_i=0$, or one of the coordinates of the vector $\mathbf k$ is negative,
\item $c_{\mathbf k}=1$ if $k_1\ge1$ and all other $k_i=0$.
\end{itemize}
We conclude from Eq. \eqref{moT} that,
\[\Expectation{U_{n}^{m}} =\sum_{\mathbf k \in \mathcal{A}_{m} }c_{\mathbf k}\prod_{i=1}^m H_{n,i}^{k_i}.\]

The technique for calculating the $m$--th derivative of the Laplace 
transform of $\tau^{(n)}$ given by Eq. \eqref{eqLaptau} is the same but requires new notation 

\begin{align*}
\hat A(y)&=\frac{1}{y-1}\cdot\frac{1}{y+2}\cdot \ldots \cdot\frac{1}{y+n},\\
\hat b_{m}(y)&=\frac{1}{(y-1)^{m}}+\frac{1}{(y+2)^{m}}+\ldots+\frac{1}{(y+n)^{m}}.
\end{align*}
Notice that $\hat A'(y)=-\hat A(y)\hat b_{1}(y)$, 
$\hat b'_{m}(y)=-m\hat b_{m+1}(y)$,
$\hat A(0)=-n!$ and \mbox{
$\hat b_{m}(0)=H_{n,m}$} if $m$ is even or $\hat b_{m}(0)=H_{n,m}-2$ if $m$ is odd.
One can then inductively show that,

\begin{align*}
\frac{\partial^{m}}{\partial y^{m}}\Expectation{e^{-y \tau^{(n)}}} & = \frac{(-1)^{m}2m!}{(n-1)(y-1)^{m-1}} -
\frac{(-1)^{m+1}(n+1)!}{n-1}\hat A(y)(\hat b_1(y)^m\\
&\quad +\sum\limits_{\substack {\mathbf k\in \mathcal{A}_{m} \\ k_1<m}}c_{\mathbf k}
\hat{\mathbf b}_m(y)^{\mathbf k}),\label{mot}\end{align*}
with the coefficients $c_{\mathbf k}$ defined as previously  by Eq. \eqref{recc}. 
Therefore, we get,

\begin{align*}
\Expectation{\tau^{(n)^{m}}} & =\frac{2m!}{n-1}-(H_{n,1}-2)^{m} + \sum_{\substack{\mathbf k \in \mathcal{A}_{m} \\ k_1<m}}c_{\mathbf k}\prod_{\substack{i=1\\i~\mathrm{odd}}}^m (H_{n,i}-2)^{k_i}\prod_{\substack{i=1\\i~\mathrm{even}}}^m H_{n,i}^{k_i}.
\end{align*}

Similarly we can use Eq. \eqref{jL} to calculate the joint moments for $U_{n}-\tau^{(n)}$ and $\tau^{(n)}$ in terms of,

\begin{align*}
A^{(i,j)}(x)&=\frac{1}{x+i+1}\cdot \ldots \cdot\frac{1}{x+j},\\
b_{m}^{(i,j)}(x)&=\frac{1}{(x+i+1)^{m}}+\ldots+\frac{1}{(x+j)^{m}}.
\end{align*}
For $m\ge1$ and $r\ge1$ we first get,

\begin{align*}
\frac{\partial^{m+r}}{\partial x^{m}\partial y^{r}}&\Expectation{e^{-x (U_{n}-\tau^{(n)})-y \tau^{(n)}}} = 
(-1)^{m+r}\frac{2(n+1)!}{n-1}\\
& \times
\sum\limits_{j=1}^{n-1}\frac{A^{(0,j)}(x) A^{(j,n)}(y)}{(j+1)(j+2)} \left(\sum\limits_{\substack {\mathbf k\in \mathcal{A}_{m}}}c_{\mathbf k}
\mathbf b_m^{(0,j)}(x)^{\mathbf k}\right)
\left(\sum\limits_{\substack {\mathbf k\in \mathcal{A}_{r}}}c_{\mathbf k}
\mathbf b_r^{(j,n)}(y)^{\mathbf k}\right),\end{align*}
and then from the above,

\begin{align*}
\Expectation{(U_{n}-\tau^{(n)})^{m}\tau^{(n)^{r}}} & = 
(-1)^{m+r}\frac{2(n+1)}{n-1}\\
&\hspace{-15mm}\times
\sum\limits_{j=1}^{n-1}\frac{1}{(j+1)(j+2)} \left(\sum\limits_{\substack {\mathbf k\in \mathcal{A}_{m}}}c_{\mathbf k}
\prod_{i=1}^m H_{j,i}^{k_i}\right)
\left(\sum\limits_{\substack {\mathbf k\in \mathcal{A}_{r}}}c_{\mathbf k}
\prod_{i=1}^r (H_{n,i}-H_{j,i})^{k_i}\right).
\end{align*}

\bibliographystyle{plainnat}
\bibliography{SagitovBartoszek}

\end{document}